\documentclass{iopart}

\usepackage{amssymb}
\usepackage{amsthm}
\usepackage{graphicx}
\usepackage{subfig}
\usepackage{color}

\newtheorem{theorem}{Theorem}[section]
\newtheorem{lemma}[theorem]{Lemma}
\newtheorem{proposition}[theorem]{Proposition}

\newtheorem{conjecture}[theorem]{Conjecture}

\begin{document}

\title{Few-Particle Vortex Cluster Equilibria in Bose-Einstein
Condensates: Existence and Stability}
\author{Anna M. Barry$^1$, PG Kevrekidis$^{1,2}$}
\address{$^1$ Institute for Mathematics and its Applications, University
of Minnesota, Minneapolis, MN 55455}
\address{$^2$ Department of Mathematics and Statistics, University of Massachusetts,
Amherst MA 01003-4515, USA}

\begin{abstract}
Motivated by recent experimental and theoretical studies of few-particle
vortex clusters in Bose-Einstein condensates, we consider the ordinary
differential equations of motion and systematically examine settings for up
to $N=6$ vortices. We analyze the existence of corresponding stationary state
configurations and also consider their spectral stability properties.
We compare our particle model results with the predictions of
the full partial differential equation system. Whenever possible,
we propose generalizations of these results in the context of
clusters of $N$ vortices. Some of these, we can theoretically
establish, especially so for the N-vortex polygons, while others
we state as conjectures, e.g. for the N-vortex line equilibrium.
\end{abstract}

\section{Introduction}

The realm of Bose-Einstein condensates (BECs) 
has offered in the past two decades
a pristine setting for the exploration of numerous nonlinear wave structures
and their interactions~\cite{pethick,stringari,emergent}. One of the most
prominent examples of such states that has received considerable attention
consists of matter wave vortices, which have by now been reviewed in 
numerous works~\cite{donnely,fetter1,fetter2,usmplb,tsubota}. Most of the
relevant experimental
studies have been concerned with various techniques of producing 
single charge vortices. Corresponding examples include, among others,
phase imprinting methods between two hyperfine states~\cite{middel7},
stirring the BEC above a critical rotation speed~\cite{middel8},
supercritical dragging of defects through the BEC~\cite{middel14,BPA10},
quenching through the condensation quantum phase 
transition~\cite{BPA08,freilich10} or the effectively nonlinear
interference of atomic BEC fragments~\cite{middel15}. In addition,
experimental efforts were also focused on producing vortices of higher
topological charge~\cite{middel16}, as well as on providing large amounts
of angular momentum with the aim of generating robust triangular vortex 
lattices~\cite{middel13}.

On the other hand, considerably less experimental effort was originally
invested on the exploration of clusters of few vortices. It was realized
early that such clusters with same circulation vortices could be 
created~\cite{middel8} and the expectation was that such states
in the presence of angular momentum would shape up into canonical
polygons~\cite{castindum} with or without a vortex located at the center.
However, in more recent experimental efforts, states with two vortices
in the form of a counter-circulating vortex 
dipole~\cite{BPA10,freilich10,middel_pra11},
as well as vortex tripoles~\cite{bagn} (with two vortices of one
sign, and one vortex of the opposite sign) have been produced
and their dynamics monitored. A very recent work has also produced
sets of 2-, 3-, 4-vortices exploring their dynamics in the absence of
a rotational angular momentum induced term~\cite{navar13}. 
These experimental works have, in turn, either had as a 
preamble~\cite{middel25,middel26,middel27,middel28,middel29,middel30,middel31,middel32,middel2010} or 
subsequently
motivated~\cite{middel24,middel_physd,middel_pla,middel_cpaa,middel_anis}
studies on the statics, stability and dynamics of such 
vortex clusters (predominantly, in fact, the vortex dipole).

Most of the above theoretical works on the vortex clusters have stemmed
from an improved understanding of the underlying partial differential
equation (PDE) which describes the pancake-shaped (i.e., quasi-two-dimensional)
condensates which contain such clusters. However, in parallel,
a growing fraction of the literature has been 
advocating~\cite{middel_pra11,navar13,middel2010,middel_cpaa,middel_anis}
the usefulness of a particle model corroborating the two principal
features of the vortex motion. These are that each of the vortices
has a precessional motion (dictated by its charge, distance from the
origin of the parabolic trap confining the BEC and characteristics of
the BEC, namely its background density at the center characterizing
the so-called
chemical potential $\mu$ and the trap confinement frequency $\Omega$)
and also has a 
relative position dependent interaction with other vortices present
in the BEC. In the limit of large chemical potentials (so-called
Thomas-Fermi limit), where the size of the vortex core shrinks effectively
to a point, hence the structure of the core plays no role in the dynamics,
it is expected that incorporating these two principal features into 
particle-based ordinary differential equations (ODEs) should be sufficient
to capture the possible vortex cluster equilibria and to assess their
(spectral) stability. This approach can, in fact, also be justified
rigorously; see for example the derivation of the relevant ODEs
in~\cite{chang}.

The aim of the present paper is to carry out the above mentioned program
to the extent that it is possible analytically and/or in
a numerically assisted form. We focus especially on the context of 
configurations that arise as stationary states from the system
with a few vortices i.e., $N=3$ up to $N=6$. We then attempt to generalize
the conclusions drawn from these low dimensional systems to the 
extent possible to larger dimensional ones carrying, however, suitable
symmetries. As principal examples, we present the case of N-gons
where the vortices occupy the vertices of a canonical polygon,
as well as that where the vortices are aligned along an axis of
the BEC. In the former, we can prove some of the relevant stability
conclusions (either analytically or in a numerically assisted form),
while in the latter, we conjecture the general result based on our
numerical observations, but leave the relevant proof as an open problem
for future study.

Our presentation is structured as follows. In section II, we briefly
present the theoretical setup, equations of motion and associated
conservation laws.
In section III, we focus on the realm of small vortex numbers $N=3, \dots
6$, while in section IV, we attempt to generalize our conclusions
to larger $N$, under suitable symmetry constraints. 
In section V, we summarize our findings 
and present our conclusions, as well as identify a number of directions for
future work.


\section{Theoretical Setup}

At the PDE level, the system of interest can be described by a two-dimensional
equation of the Nonlinear Schr{\"o}dinger (NLS) or Gross-Pitaevskii (GPE)
type, where the trap strength is parametrized by an effective frequency
$\Omega_{eff}=\omega_r/\omega_z$ (i.e., the ratio between the in-plane and
perpendicular to the plane trapping frequencies) and the density (at the center
of the trap) by the chemical potential parameter $\mu$, directly associated
with the number of atoms in the BEC. Details about the PDE level description
and the reduction from the original three-dimensional (3d) system
to the effective two-dimensional (2d) one can be found in~\cite{middel_physd}
for our setup.
In some sense, our work will naturally complement the above manuscript,
as the latter considers how vortex cluster states emerge from the
linear limit of the quantum harmonic (2d) oscillator and bifurcations
of nonlinear states from linear states thereof (i.e., regime of
small chemical potential). Here, on the other hand, we will
focus on the opposite limit of large chemical potential and
the particle system emerging when considering the vortices as
isolated precessing and interacting entities characterized by
their position within the two-dimensional plane.

In this spirit, let $\mathbf{x}_j=(x_j,y_j)$ be a point vortex in the BEC system with charge $S_j$ for $j=1,...,N$ and consider the system
\begin{eqnarray}
\dot{x}_j&=-S_j\Omega y_j-\frac{b}{2}\sum_{k\neq j}^N S_k\frac{y_j-y_k}{|\mathbf{x}_j-\mathbf{x}_k|^2}
\label{eqn1}
\\
\dot{y}_j&=S_j\Omega x_j+\frac{b}{2}\sum_{k\neq j}^N S_k\frac{x_j-x_k}{|\mathbf{x}_j-\mathbf{x}_k|^2}
\label{eqn2}
\end{eqnarray}
where we fix $b=2$ for simplicity~\footnote{Our considerations herein
will not be significantly 
affected by the precise value of $b$, as long as the latter assumes
a constant value proximal to the one of the homogeneous limit that
we assume here. We should note, however, for completeness that to
improve the agreement between the ODE as regards the precise location
of the fixed points, the work of~\cite{middel2010} and others thereafter,
considered an effective value of $b=1.35$. This was intended to 
account for the density-induced ``screening'' effect associated with
the vortex interaction. A first-principles accounting of such screening
would necessitate the study of integro-differential equations 
as analyzed in~\cite{mcendoo} [see Eq. (21) therein].}.
$\Omega$ in Eqs.~(\ref{eqn1})-(\ref{eqn2}) is the precession frequency
of a single vortex in a trap, which is well-known to depend
on the effective trap frequency $\Omega_{eff}$ and the number of
atoms in the BEC (as characterized by the so-called chemical 
potential)~\cite{freilich10,middel2010}.

The above system is Hamiltonian with 
\begin{equation}
H(x_1,y_1,...,x_N,y_N)=-\Omega \sum_{j=1}^N S_j r_j^2+\sum_{j< k}^N S_j\log(r_{jk}^2)
\label{eqn3}
\end{equation}
where $r_{jk}=|z_j-z_k|$, and $z_j$ is the complex variable $z_j=x_j+iy_j$.  
In terms of this variable, the equations reduce to
\begin{equation}
i\dot{z}_j=-S_j\Omega z_j+\sum_{k\neq j}^N \frac{S_k}{\bar{z}_j-\bar{z}_k}.
\label{eqn4}
\end{equation}
We also note in passing that that the angular momentum $L= \sum 
S_j r_j^2$ is also a conserved quantity for the system of vortices. We use
the term angular momentum for this conservation law in line with the tradition
stemming from the literature of point vortices in fluid 
mechanics~\cite{aref,newton01} (rather than the angular momentum
of the full quantum mechanical problem).
Finally, it will be useful to cast 
the system in polar coordinates $(r_j,\theta_j)$, in which case
the equations of motion become:
\begin{eqnarray}
\dot{r}_j&=\sum_{k\neq j}\frac{S_k r_k \sin(\theta_k-\theta_j)}{r_{jk}^2}
\label{eqn5}
\\
r_j\dot{\theta}_j&=r_jS_j\Omega+\sum_{k\neq j}\frac{S_k}{r_{jk}^2}\left(r_j-r_k\cos(\theta_j-\theta_k)\right).
\label{eqn6}
\end{eqnarray}

It should be noted here that in the present work the precession frequency
of a single vortex will be assumed to be constant, an assumption that
is fairly accurate between the center of the trap confining the BEC
and half of its radial extent (the latter is often referred to as
the Thomas-Fermi radius). More general position dependent precession
frequency expressions can be used also in connection to
experiments, such as most notably
$\tilde{\Omega}_j=\Omega/(1-r_j^2)$~\cite{freilich10,middel_pra11,navar13}.
However we have checked that
these do not modify the existence or stability conclusions for the
solutions presented in this paper. Hence, for simplicy of the exposition,
we restrict our presentation to the constant $\Omega$ case hereafter.

\section{Small $N$}

We now focus more specifically on the case of small vortex clusters. 
It will be clear from what follows that the cases of interest to us
will be those where the vortices do not all carry the same charge.
In the latter case, the rotation imposed by the precession is only
enhanced by the rotation induced by the interaction and hence
the vortices cannot find themselves in a situation of genuine
equilibrium. Rather, in the latter case, one can only talk about 
rigidly rotating states as was examined in the recent work of~\cite{navar13}
for $N=2$, $3$ and $4$. For large $N$, the latter setting has been
examined too, with a recent example being the work of~\cite{kolok};
see also therein for relevant references.
Here, on the other hand, we deal with genuine equilibria of
the vortex cluster system and hence none of our configurations 
carry vortices of a single charge type.

\subsection{$N=3$}

The case $N=2$ was examined in detail in \cite{middel_pra11,middel_pla}, 
in which the dipole was found to be the only fixed point and is linearly stable.
Hence, we start by focusing our considerations to the case of $N=3$.

\begin{proposition} When $N=3$, the only fixed point of the particle system is a collinear configuration.\end{proposition}

\begin{proof}  Without loss of generality, fix $y_1=0$.  Suppose $(x_1,0,x_2,y_2,x_3,y_3)$ with charges $S_1$, $S_2$ and $S_3$ is a fixed point of the system.  Fix $S_1=1$.  Under these assumptions, the equations read
\begin{eqnarray}
\frac{S_2 y_2}{r_{12}^2}&=\frac{-S_3 y_3}{r_{13}^2}
\label{eqn7}
\\
\frac{S_3(y_2-y_3)}{r_{23}^2}&=-S_2\Omega y_2-\frac{y_2}{r_{12}^2}
\label{eqn8}
\\
\frac{S_2(y_3-y_2)}{r_{23}^2}&=-S_3\Omega y_3-\frac{y_3}{r_{13}^2}
\label{eqn9}
\end{eqnarray}
There are three possible cases to consider:  $S_2=S_3=\pm1$, $S_2=-S_3=1$. In each case, we obtain the relation $y_2=-y_3$.  For instance, consider $S_2=-S_3=1$.  Then,
\begin{eqnarray*}
\frac{y_2}{r_{12}^2}&=\frac{y_3}{r_{13}^2}\\
\frac{y_3-y_2}{r_{23}^2}&=-\Omega y_2-\frac{y_2}{r_{12}^2}\\
\frac{y_3-y_2}{r_{23}^2}&=\Omega y_3-\frac{y_3}{r_{13}^2}.
\end{eqnarray*}
Thus inserting the first relation into the second equation and equating the second and third equations yields $y_2=-y_3$.  \newline
\indent  Note that even though we have satisfied the equations for equilibrium in x, it is still possible $\mathbf{\dot{y}}\neq0$.  To see that a true collinear configuration exists, fix $\Omega=1$ and set $S_1=1$, $S_2=S_3=-1$.  One can check that the configuration $\mathbf{x}_1=(0,0)$, $\mathbf{x}_2=(0,\frac{\sqrt{2}}{2})$ and $\mathbf{x}_3=(0,-\frac{\sqrt{2}}{2})$ satisfies 
the fixed point equations.
\end{proof}

One can then check by a direct calculation that the collinear fixed point is linearly unstable.  The eigenvalues of the linearization matrix are $\lambda_{1,2}=\pm i\sqrt{5}$, $\lambda_{3,4}=\pm\sqrt{7}$ and $\lambda_{5,6}=0$.  In general we expect the collinear configuration to have $N-2$ real directions of instability. This is consonant with the conclusions of~\cite{middel2010,middel_physd}.

\subsection{$N=4$}

The case $N=4$ is more subtle, as the system exhibits more than one fixed 
points, as we show analytically below.  Unlike the classical point vortex problem, the system (1)-(2) does not exhibit translational symmetry and therefore the ``center of vorticity'' is not conserved.  In fact, 
\begin{eqnarray}
\sum_{j=1}^N S_j \dot{x}_j&=-\Omega\sum_{j=1}^N y_j\label{eq:3}\\
\sum_{j=1}^N S_j \dot{y}_j&=\Omega \sum_{j=1}^N x_j\label{eq:4}
\end{eqnarray}
which implies that the fixed points sought herein must
have a center of mass located at $(0,0)$.
Moreover, the system exhibits rotational symmetry which is clear by replacing $\theta\mapsto \theta+\omega t$.  We make the following conjecture.  

\begin{conjecture} All fixed points of (\ref{eqn1})-(\ref{eqn2}) 
are symmetric about the origin.\end{conjecture}

Assuming the conjecture is true, one can prove 

\begin{proposition} For $N=4$, the only fixed points of (1)-(2) are square or collinear (i.e., one in which all the vortices are located on a line
going through the origin) configurations.\end{proposition}

\begin{proof}  Choosing $\mathbf{x}_1=(x_1,0)$ implies that $\mathbf{x}_2=(-x_1,0)$.  Then (\ref{eq:3})-(\ref{eq:4}) imply $x_3=-x_4$ and $y_3=-y_4$.  If $y_3=0$, the configuration is collinear, 
so assume $y_3\neq 0$.  Then the equation for $\dot{x}_1$ yields
\begin{equation*}
\frac{S_3}{r_{13}^2}=\frac{S_4}{r_{14}^2}.
\end{equation*}
Thus $r_{13}=r_{14}$, but since the configuration is symmetric about the origin this implies that it is a square.\end{proof}

Fixing $\mathbf{x}_1=(1,0)$, $\mathbf{x}_2=(-1,0)$, $\mathbf{x}_3=(0,1)$, $\mathbf{x}_4=(0,-1)$, $S_1=S_2=1$, and $S_3=S_4=-1$ gives a square fixed point with $\Omega=\frac{1}{2}$.  Linearizing about this equilibrium yields eigenvalues
\begin{equation}
\lambda_{1,2}=\pm \frac{i}{\sqrt{2}},\; \lambda_{3,4}=\lambda_{5,6}=\pm \frac{i}{4},\; \lambda_{7,8}=0
\end{equation}
with corresponding eigenvectors 
\begin{eqnarray*}
v_1&=\{1/3 i (i + 2 \sqrt{2}), 1/3 (1 - 2 i \sqrt{2}), 1/3 (1 - 2 i \sqrt{2}), 
 1/3 i (i + 2 \sqrt{2}), -1, -1, 1, 1\}=\bar{v}_2\\
v_3&=\{0, i, 0, i, 0, 1, 0, 1\}=\bar{v}_4\\
v_5&=\{-i, 0, -i, 0, 1, 0, 1, 0\}=\bar{v}_6\\
v_7&=\{-1, -1, 1, 1, 1, -1, -1, 1\}\\
v_8&=\mathbf{0}.
\end{eqnarray*} 

Fixing $\Omega=\frac{1}{2}$ gives a collinear fixed point $\mathbf{x}_1=(c,0)=-\mathbf{x}_3$, $\mathbf{x}_2=(d,0)=-\mathbf{x}_4$ where
\begin{eqnarray*}
c&=\sqrt{1 + \sqrt{2 (-1 + \sqrt{2})}}\\
d&=\frac{1}{2}(10 \sqrt{1 + \sqrt{2 (-1 + \sqrt{2})}} - 
   9 (1 + \sqrt{2 (-1 + \sqrt{2})})^{3/2} \\
&+ 4 (1 + \sqrt{2 (-1 + \sqrt{2})})^{
    5/2} - (1 + \sqrt{2 (-1 + \sqrt{2})})^{7/2}).
\end{eqnarray*}

This configuration has eigenvalues $\lambda_{1,2}\approx\pm 2.15$, $\lambda_{3,4}\approx\pm 1.43 i$, $\lambda_{5,6}\approx \pm.914$, $\lambda_{7,8}=0$, and 
hence possesses $N-2=2$ directions of instability. 
\begin{figure}[h]
\centering
\includegraphics[width=4in]{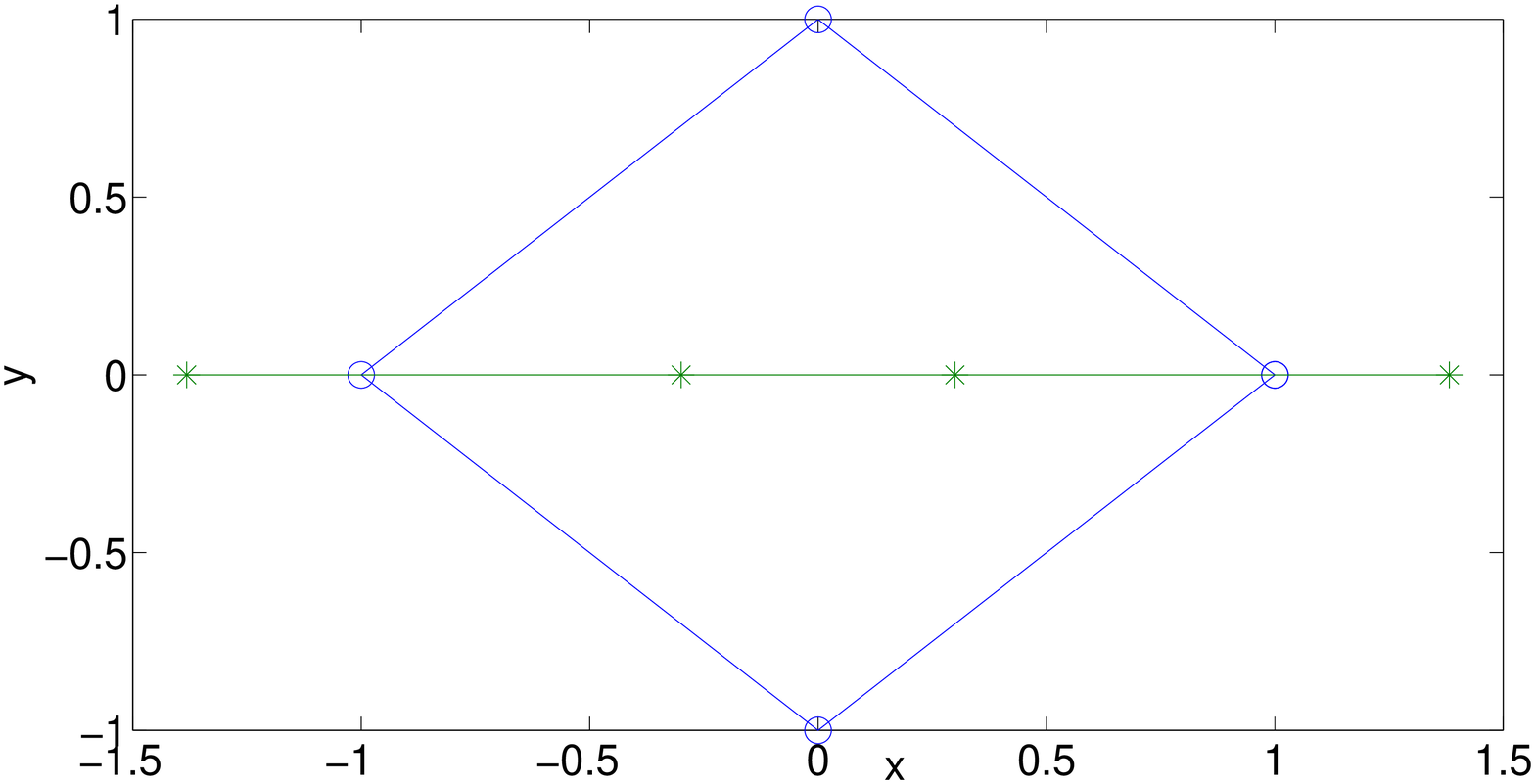}\;\;
\caption{Fixed points of the system in the case of $N=4$ and
with $\Omega=\frac{1}{2}$. The square configuration has the vortices
denoted by circles, while the collinear one has them denoted by stars.
Notice that adjacent vortices have opposite sign (unit magnitude) charge.}
\end{figure}

These results are consonant with the corresponding findings of
the 2d PDE of the NLS/GPE type~\cite{middel2010,middel_physd}.
For the square configuration the PDE predicts that such a state
exists even at the linear limit and bifurcates therefrom as a linearly
stable equilibrium (see also~\cite{rotating}) with the exception of
a possibility for an oscillatory instability (Hamiltonian Hopf bifurcation)
which can arise for an intermediate range of values of $\mu$. Nevertheless,
the state is linearly stable in the limit of large $\mu$. It should be
highlighted that this state also has a neutral direction because of
its radial symmetry within the isotropic two-dimensional parabolic
trap. Such a direction is shared also by the collinear configuration
which also has a pair of zero eigenvalues due to its invariance 
under rotations of its linear axis. Moreover, the latter configuration
has 2 directions of instability, in line with the expectation that
each higher collinear configuration will have an additional real eigenvalue
pair (than the previous one -- starting with the dipole that has none,
then the tripole has one, the aligned quadrupole two, etc.). This
expectation stems from the supercritical nature of the bifurcation
of these collinear states from the dark soliton state as
discussed in~\cite{middel2010,middel_physd}.

\subsection{$N=5$}
In the next section we prove that for general $N$ odd the $N$-gon 
(i.e., the canonical polygon with the vortices at its vertices)
is not a fixed point.  We thus hereafter investigate other configurations,
confining our considerations to ones 
that are symmetric about the origin.  Fix $(x_1,y_1)=(0,0)$ and $(x_2,y_2)=(r_1,0)$.  Then by symmetry $(x_3,y_3)=(-r_1,0)$.  Set $(x_4,y_4)=(r_2 \cos\theta_2,r_2\sin\theta_2)$, $0\leq \theta_2<\pi$, so that $(x_5,y_5)=(-r_2 \cos\theta_2,-r_2\sin\theta_2)$.  We compute
\begin{equation*}
\dot{x}_1=\frac{(S_3-S_5)\sin\theta_2}{r_2}
\end{equation*}
and so either $\sin\theta_2=0$ or $S_3=S_5$.  Assuming the latter holds we find from the equation for $y_1$ that $S_2=S_4$.  Then from the equation for $x_2$ it follows that $\sin(2\theta_2)=0$.  
If we choose the root $\theta_2=\frac{\pi}{2}$, then the configuration must be a square or rhombus centered on the origin.  Here, two cases are physically relevant.  One where the vortices alternate in sign in the counterclockwise direction with the central vortex being of either sign, and another where the outer vortices have the same sign opposite to the sign of the vortex at the origin.  In the first case, we find that the system has no fixed points.  In the second, we find that if the central vortex has charge $M>0$ and the outer vortices have charge $-1$, then the system has a fixed point if and only if $r_1=r_2$.  As an example, when the radius is one and $M=2$ we compute the configuration to be unstable with eigenvalues $\lambda_{1,2}\approx 3.14 i$, $\lambda_{3,4}=\sqrt{6}$, $\lambda_{5,6}\approx -1.15\pm 1.07 i$, $\lambda_{7,8}\approx 1.15\pm 1.07 i$ and $\lambda_{9,10}=0$.
More generally, considering the algebraic constraints, one finds
that the configuration with a square surrounding the central vortex
can be realized provided suitable contraints connecting the central
vortex charge, the surrounding charges,
the precession frequency and the square's radius. The relevant
condition reads 
\begin{eqnarray}
r_1^2 = \frac{2 M -3 \tilde{M}}{2 \Omega \tilde{M}}
\label{extra1}
\end{eqnarray}
where $-\tilde{M}$ are the surrounding charges to the central one.
It is thus clear that such configurations will only exist
if $\tilde{M} \times (2 M -3 \tilde{M}) > 0$ and assuming
$\tilde{M}>0$ without loss of generality, this leads to $M > 3 \tilde{M}/2$.
In the case of $\tilde{M}=1$, the lowest charge that will work is $M=2$.
This configuration can be found to be definitely 
unstable due to a real pair $\lambda= \pm 2 \tilde{M}^{3/2} \Omega\sqrt{4 M -2 \tilde{M}}/{(2 M-3
\tilde{M})}$,
while other eigenvalues have more complicated functional forms
not provided here. As in the above numerical example, we find
a quartet of complex eigenvalues, a real pair, an imaginary pair and
a neutral pair  associated with the rotational invariance of
such a 5-vortex state.
There are two observations to make here in connection to this.
This configuration is the same as the one labeled ``5x'' in the work
of~\cite{middel_physd}~\footnote{However, note that inadvertently the
latter work mentioned a quadrupole as surrounding the central vortex
in p. 1453; the correct statement is that 4 same charge vortices,
opposite in sign to the doubly charged central one are surrounding it.}.
The second is that in line with the numerical observations 
of~\cite{middel_physd}, we find that this configuration at large 
$\mu$ contains a complex eigenfrequency quartet, as well as a real
eigenvalue pair as manifestations of its instability~(see accordingly
the 3rd row, right panel of Fig.~4 in p.~1453 of~\cite{middel_physd}).

\indent We now turn to the case when $\theta_2=0$; here, we have a collinear fixed point which only exists when the vortices have alternating charges.  
This fixed point can be described by the configuration $(0,0)$, $(\pm r_1,0)$, $(\pm r_2,0)$ where
\begin{eqnarray*}
r_1&=\sqrt{\frac{2 - \sqrt{3}}{2 \Omega}}\\
r_2&=\sqrt{\frac{\sqrt{3}}{2 \Omega}}
\end{eqnarray*}
This configuration was computed to be unstable with eigenvalues $\lambda_{1,2}\approx \pm 6.95 \Omega$, $\lambda_{3,4}\approx \pm 6.69 \Omega$, 
$\lambda_{5,6}\approx \pm 3.64 i \Omega$, $\lambda_{7,8}\approx \pm 2.8 \Omega $, $\lambda_{9,10}=0$ and so again we see that there are $N-2=3$ directions of real instability, as well as a neutral
direction. This is in line with the expectations of the earlier
works of~\cite{middel2010,middel_physd}, examining the PDE limit of
such 5-vortex collinear configurations.

\begin{figure}[h]
\centering
\includegraphics[width=4in]{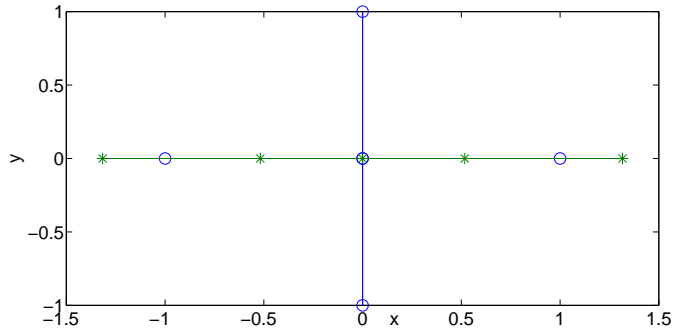}
\caption{The prototypical fixed points that exist for $\Omega=\frac{1}{2}$
and $N=5$. The circles denote the example where an $M=2$ vortex at the origin
is surrounded by $4$ vortices of charge 
$-\tilde{M}=-1$ over the nodes of
a square each at a unit distance from the origin. The stars denote the 
collinear configuration of alternating sign (unit magnitude) charges.}
\end{figure}

\subsection{$N=6$}

We are aware of two fixed points when $N=6$: the hexagon and the collinear configuration.  
Both of these are configurations of alternating charges (either along the
ring or along a line, respectively).
The linearization about the hexagon when all vortices are on the unit circle and $\Omega=\frac{1}{2}$ yields the
following eigenvalues
\begin{eqnarray*}
\lambda_{1,2}&=\pm \frac{3 i}{\sqrt{2}},\; \lambda_{3,4}=\lambda_{5,6}=\pm i, \; \lambda_{7,8}=\lambda_{9,10}=\pm \frac{1}{\sqrt{2}},\;\lambda_{11,12}=0
\end{eqnarray*}
and so, unlike the square, the hexagon is unstable.  
This is in line with earlier observations at the PDE level for this vortex
ring, see e.g. the discussion of~\cite{middel_physd} (top right of
Fig. 5 in p. 1454
and associated discussion). There, it is inferred (coming from the
opposite limit of small chemical potential) that the hexagon supercritically
bifurcates from the already unstable (even off of and near the linear limit)
state of the dark soliton ring~\cite{djf}, inheriting its instability.

\indent The collinear configuration can be represented by $(\pm a,0)$, $(\pm b,0)$, $(\pm c,0)$ where $a\approx 0.23$, $b\approx 0.67$ and $c\approx 1.58$ were computed numerically.  The eigenvalues of the linearization about this fixed point were computed to be $\lambda_{1,2}\approx\pm 4.88$, $\lambda_{3,4}\approx \pm 4.86$, $\lambda_{5,6} \approx\pm 2.65$, $\lambda_{7,8}\approx \pm 2.2i$, $\lambda_{9,10}\approx\pm 0.96$, and $\lambda_{11,12}=0$ and so this configuration also has $N-2=4$ directions of instability. We note in this case too
that in order to generalize the relevant results in the case of arbitrary
$\Omega$, one has to scale the fixed point positions by $1/\sqrt{\Omega}$
and the corresponding eigenvalues linearly by $\Omega$.

\begin{figure}[h]
\centering
\includegraphics[width=4in]{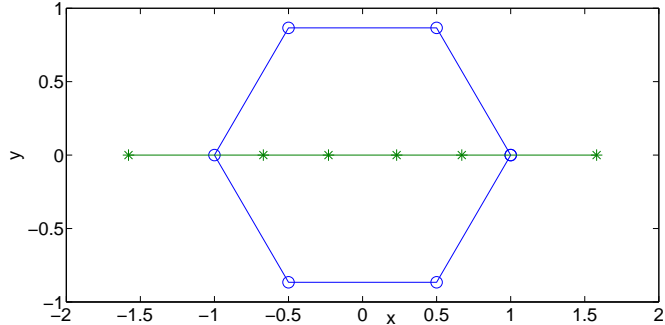}
\caption{Two fixed points when $N=6$ and again for $\Omega=\frac{1}{2}$.
The hexagonal configuration is shown by (blue) circles, while the 
collinear one by (green) stars. We again note as in the case of $N=4$
that adjacent vortices have opposite charges between each other.}
\end{figure}

\section{Large $N$}
\subsection{The $N$-gon fixed point}


Up to now, we have provided detailed (existence and
stability) features of setups containing
only small numbers of vortices.
We now generalize our conclusions, by offering some proofs, as well as
some conjectures for more general vortex configurations featuring
large numbers of vortices $N$. As we will see, there are some significant
differences between the cases of an even number of vortices and those
of an odd number of vortices, both at the level of existence, as well
as at that of stability.  We will assume hereafter that the vortices
are of alternating charge, namely $S_j=(-1)^{j+1}$, $j=1,2, \dots,N$.

\begin{proposition} The polygonal $N$-vortex (hereafter, sometimes
referred to as ``$N$-gon'') configuration 
described above is a fixed point of the point vortex system 
if and only if $N$ is even. Moreover, the configuration is linearly 
stable for $N=4$ and linearly unstable otherwise. \end{proposition}

First, if the $N$-gon is an equilibrium any radial scaling is also an equilibrium for some different choice of $\Omega$, and so we may assume that each vortex lies on the unit circle.  This leaves a single free real parameter $\Omega$, and $z_1= e^{2\pi i/N}, z_2=e^{4\pi i/N},..., z_N=1$.  We have examined the  
$N=4$ and $N=6$ (i.e., the marginally stable and the first
unstable) cases in detail above.  Notice that by abusing notation,
we can also consider the dipole as an example of this type with $N=2$,
which is stable at the particle level 
in accordance with the earlier analysis e.g. of~\cite{middel_pla}.

The first proposition relies on the following two lemmas.

\begin{lemma} If $N$ is odd, the $N$-gon is not a fixed point. If $N$ is even, the $N$-gon is a fixed point.\end{lemma}

\begin{proof} This is easiest to see in polar coordinates.  First assume that the point vortices lie on the unit circle.  When $N$ is even, note that by symmetry $\dot{r}_j=0$ for each $j$. For $\mathbf{r}=\mathbf{1}$, one finds
\begin{eqnarray*}
\dot{\theta_j}&=(-1)^{j+1} \Omega+\frac{1}{2}\sum_{k\neq j} (-1)^{k+1}\\
&=\Omega-\frac{1}{2},\;\;\mathrm{if}\; j \;\mathrm{is}\; \mathrm{odd}\\
&=-\Omega+\frac{1}{2},\;\; \mathrm{if}\; j \;\mathrm{is}\; \mathrm{even}
\end{eqnarray*}
and so the $N$-gon is a fixed point when $\Omega=1/2$.  On the other hand, when $N$ is odd we have
\begin{eqnarray*}
\dot{\theta_j}&=(-1)^{j+1} \Omega+\frac{1}{2}\sum_{k\neq j} (-1)^{k+1}\\
&=\Omega,\;\;\mathrm{if}\; j \;\mathrm{is}\; \mathrm{odd}\\
&=-\Omega+1,\;\; \mathrm{if}\; j \;\mathrm{is}\; \mathrm{even}
\end{eqnarray*}
and so for no value of $\Omega$ is the odd $N$-gon a fixed point.

\end{proof}

\begin{lemma} When the $N$-gon is a fixed point, it is unstable for $N\geq 6$.\end{lemma}
\begin{proof}
This result will require a detailed examination of the corresponding stability
matrix.
We again consider the system in polar coordinates.  The resulting matrix, denoted by $M$, of the linearization is a $2N\times2N$ block matrix made up of $N\times N$ blocks

\[M:= \left( \begin{array}{cc}
A & B \ \\
C & D\end{array} \right).\]

We compute each of these blocks explicitly.  Let $a_{j,k}$ denote the entry of $A$ in the $j^{th}$ row and $k^{th}$ column, and similarly define $b_{j,k}$, $c_{j,k}$ and $d_{j,k}$.  Let $\theta_{jk}=\theta_j-\theta_k$ and $r_{jk}=|z_j-z_k|$.  Then
\begin{eqnarray*}
a_{j,k}&=\frac{S_j (r_k^2-r_j^2) \sin(\theta_{jk})}{r_{jk}^4},\; j\neq k\\
a_{j,j}&=\sum_{k\neq j} \frac{2 r_k S_k (r_j-r_k\cos(\theta_{jk}))\sin(\theta_{jk})}{r_{jk}^4}.
\end{eqnarray*}
In the case of the $N$-gon where $S_k=(-1)^{k+1}$ and $r_j=1$ for all $j$ we see that the above entries are identically zero, hence $A=\mathbf{0}$, the zero matrix.  Next,
\begin{eqnarray*}
b_{j,k}&= \frac{r_j S_j((r_j^2+r_k^2)\cos(\theta_{jk})-2r_j r_k)}{r_{jk}^4},\; j\neq k\\
b_{j,j}&=\sum_{k\neq j} \frac{r_k S_k(2 r_k r_j-(r_j^2+r_k^2)\cos(\theta_{jk}))}{r_{jk}^4}.
\end{eqnarray*}
In this case, when all radii are set to one we find
\begin{eqnarray*}
b_{j,k}&= -\frac{S_j}{r_{jk}^2}\\
b_{j,j}&= \sum_{k\neq j} \frac{S_k}{r_{jk}^2}.
\end{eqnarray*}
The matrix $C$ is
\begin{eqnarray*}
c_{j,k}&= \frac{S_k((r_j^2+r_k^2)\cos(\theta_{jk})-2 r_j r_k)}{r_k r_{jk}^4},\; j\neq k\\
c_{j,j}&= -\sum_{k\neq j} \frac{S_k (r_j^2-2r_j r_k \cos(\theta_{jk})+r_k^2 \cos(2\theta_{jk}))}{r_j r_{jk}^4}-\sum_{k\neq j} \frac{S_k (r_j-r_k \cos(\theta_{jk}))}{r_j^2 r_{jk}^2}.
\end{eqnarray*}
Again setting $\mathbf{r}=1$ one obtains after some simplification
\begin{eqnarray*}
c_{j,k}&=-\frac{S_j}{r_{jk}^2}\\
c_{j,j}&=\sum_{k\neq j} \frac{S_k}{r_{jk}^2}-\sum_{k\neq j} S_k.
\end{eqnarray*}
Here we have used that $r_{jk}^2=2-2\cos(\theta_{jk})$ when $\mathbf{r}=\mathbf{1}$.
Finally, we have
\begin{eqnarray*}
d_{j,k}&= \frac{S_j r_j(r_j^2-r_k^2)\sin(\theta_{jk})}{r_k r_{jk}^4},\; j\neq k\\
d_{j,j}&=\sum_{k\neq j} \frac{S_k r_k (r_k^2-r_j^2)\sin(\theta_{jk})}{r_j r_{jk}^4}
\end{eqnarray*}
which is the zero matrix when $\mathbf{r}=\mathbf{1}$.
We are interested in the determinant of the matrix
\[ \left( \begin{array}{cc}
-\lambda I & B \ \\
C & -\lambda I \end{array} \right).\]

Since the two matrices in the bottom row of $M$ commute, one has the formula
\begin{eqnarray*}
\det(M)=\det(AD-BC)
\end{eqnarray*}
and
\begin{eqnarray*}
\det(M-\lambda I_{2N\times 2N})=\det(\lambda^{2} I_{N\times N}-BC).
\end{eqnarray*}
When $\Omega=\frac{1}{2}$, $BC$ is a circulant matrix and hence its eigenvalues are real.  Therefore, $\lambda$ must be real or purely imaginary.  To determine the eigenvalues of a circulant matrix, it is sufficient to know the first row of the matrix.  Let $\widetilde{B}=BC$.  Then
\begin{eqnarray*}
\tilde{b}_{1,1}&=\left(\sum_{k\neq 1} \frac{S_k}{r_{1k}^2}\right)^2-\sum_{k\neq 1} S_k \sum_{k\neq 1} \frac{S_k}{r_{1k}^2}+\sum_{k\neq 1} \frac{S_1 S_k}{r_{1k}^4}\\
\tilde{b}_{1,2}&=\left(\frac{S_1}{r_{12}^2}\right)\left[\sum_{k\neq 1} \frac{S_k}{r_{1k}^2}+\sum_{k\neq 2} \frac{S_k}{r_{2k}^2}-\sum_{k\neq 2} S_k\right]+\sum_{k\neq 1,2} \frac{S_1 S_k}{r_{1k}^2 r_{2k}^2}\\
\tilde{b}_{1,3}&=\left(-\frac{S_1}{r_{13}^2}\right)\left[\sum_{k\neq 1} \frac{S_k}{r_{1k}^2}+\sum_{k\neq 3} \frac{S_k}{r_{3k}^2}-\sum_{k\neq 3} S_k\right]+\sum_{k\neq 1,3} \frac{S_1 S_k}{r_{1k}^2 r_{3k}^2}\\
\end{eqnarray*}
and in general for $m\neq 1$
\begin{eqnarray}
\tilde{b}_{1,m}&=(-1)^{m+1}\frac{S_1}{r_{1m}^2}\left[ \sum_{k\neq 1} \frac{S_k}{r_{1k}^2}+\sum_{k\neq m} \frac{S_k}{r_{mk}^2}-\sum_{k\neq m} S_k \right]+\sum_{k\neq 1,m} \frac{S_k}{r_{1k}^2 r_{mk}^2}.
\end{eqnarray}
When $S_k=(-1)^{k+1}$, we have
\begin{eqnarray*}
\tilde{b}_{1,1}&= \left(\sum_{k\neq 1} \frac{(-1)^{k+1}}{r_{1k}^2}\right)^2+ \sum_{k\neq 1} \frac{(-1)^{k+1}}{r_{1k}^2}+\sum_{k\neq 1} \frac{(-1)^{k+1}}{r_{1k}^4}\\
\tilde{b}_{1,m}&=(-1)^{m+1}\frac{1}{r_{1m}^2}\left[ \sum_{k\neq 1} \frac{(-1)^{k+1}}{r_{1k}^2}+\sum_{k\neq m} \frac{(-1)^{k+1}}{r_{mk}^2}+ (-1)^{m+1} \right]+\sum_{k\neq 1,m} \frac{(-1)^{k+1}}{r_{1k}^2 r_{mk}^2}.
\end{eqnarray*}
The eigenvalues of $\widetilde{B}$ are given by 
\begin{equation*}
\gamma_j=\sum_{k=1}^{N} \tilde{b}_{1,k}\omega_j^{k-1},\; j=1,...,N
\end{equation*}
and the eigenvectors are
\begin{equation*}
v_j=\{1,\omega_j,\omega_j^2,...,\omega_j^{N-1}\}
\end{equation*}
 where $\omega_j=e^{2\pi i j/N}$ is the $j^{th}$ root of unity.  

In particular, we find that for chosen even $N$ we can explicitly compute any eigenvalue although a simplified expression for the eigenvalues was not found.  For $j=N-2$, $\gamma_j$, the eigenvalues of $\tilde{B}$, will be nonnegative when $N\geq 4$. Using Mathematica, we explicitly compute $\gamma_{N-2}=\left(\frac{\sqrt{2}}{4}N-\sqrt{2}\right)^2\geq 0$.  Thus $\lambda_j=\pm \sqrt{\gamma_j}$ is 
real and nonzero when $N>4$, and the corresponding $N$-gon is unstable.
By the same token, the $N$-gon with $N=4$ will be marginally stable.
\end{proof}

We point out here that these results, as well as numerical results
both at the level of the particle equations and at that of the
underlying PDE~\cite{middel_physd} suggest the following.
\begin{conjecture} For $N\geq 6$, the $N$-gon has $\frac{N}{2}-2$ 
distinct pairs of real eigenvalues.\end{conjecture}
At the ODE level, this stems from direct observations for
$N=2$, $N=4$, $N=6$ and $N=8$, while at the PDE, the relevant
observation is that each higher order vortex polygon
(for even $N\geq 6$) emerges from subsequent supercritical 
pitchfork bifurcations
of the unstable ring dark soliton~\cite{djf} 
and hence each additional destabilization
adds a real eigenvalue pair to the linear spectrum~\cite{middel_physd}
[see Fig. 5, top right and bottom left, p. 1454]. The $N=4$
state on the other hand emerges as a spectrally stable
one from the linear limit. Notice, however, that the stability identified
here arises at the large $\mu$ limit. For intermediate values of $\mu$
(i.e., for an intermediate interval thereof) instabilities of oscillatory
type may arise due to the PDE nature of the system, i.e., due to collisions
of some of the internal modes of the vortex particle system with those
of its host background.

\subsection{The Collinear Fixed Point}

We now briefly discuss the generalization of the features 
of the collinear
fixed point that we have numerically observed for larger numbers
of vortices. 

\begin{conjecture} For all $N$, there exists a symmetric, collinear fixed point of the equations of motion for alternating sign vortex configurations. 
For even $N$ such configurations are symmetrically placed around the
origin i.e., $(\pm a,0)$, $(\pm b, 0)$, $(\pm c, 0)$ $\dots$, while
for odd $N$ they have the same structure plus an additional vortex
placed at the origin.
\end{conjecture}

The existence of such a configuration for any $N$, suggests that
there are always real numbers $x_1$, $x_2$, ..., $x_N$ 
such that for every $j=1,...,N$, 
$\sum_{k=1}^N  (-1)^{k-j} / (x_j (x_k-x_j)) = 1$
where in the summation, it is implied that $k \neq j$.
In the case of N even, the statement is precisely that 
(and the solution, as our numerical results indicate,
 has the vortices symmetrically placed around $0$,
i.e., $-x_{N/2}, -x_{(N-2)/2}, ..., -x_1, x_1, ...,x_{(N-2)/2}, x_{N/2}$). 
In the case of N odd, one of the vortices is always placed at $0$, so
one can rephrase the statement as: 
$\sum_{k=1}^N  (-1)^{k-j} /  (x_k-x_j) = x_j$. 
The remaining vortices are placed symmetrically around
$0$ as above. It is clear that this is essentially an 
algebraic/number-theoretic
problem. Effectively, this can be rewritten as a set of N polynomial 
equations in $N$ unknonwns. As such, it defines a variety, and the
 question is whether this variety always has a real point 
$(x_1, ..., x_N)$~\cite{farshid}. 

In addition, we can generalize our conclusions for the stability of
the collinear vortex state as follows.

\begin{conjecture} The collinear fixed point is unstable for $N\geq 3$ with 
$N-2$ real directions of instability.\end{conjecture}

Earlier studies have explored the vortex dipole of $N=2$
at the particle~\cite{middel_pla} and PDE~\cite{middel2010} level.
Here, we have studied in a numerically assisted way the cases of
$N=3,\dots, 6$ at the ODE level and such small $N$ cases have
also been considered at the PDE level~\cite{middel2010,middel_physd},
all corroborating the above conclusions. Moreover, at the PDE level,
the supercritical pitchfork bifurcation of the rectilinear dark soliton
state emanating at each subsequent order a higher collinear vortex
configuration clearly signals an agreement with the above statement.
Nevertheless, the nature of the corresponding stability matrix is
not circulant and we are presently unable to prove such a statement
in its full generality although our numerical computations (even with
higher $N$) fully confirm it.

\section{Conclusions and Future Challenges}

In the present work, we revisited the topic of few vortex crystal
configurations in two-dimensional Bose-Einstein condensates.
While earlier studies focused on the PDE approach attempting to infer
conclusions for the vortex dynamics from the vicinity of the linear
limit~\cite{middel25,middel26,middel27,middel28,middel29,middel30,middel31,middel32,middel2010,middel_physd}, our approach here has taken a complementary
view whereby the vortices have been examined as interacting particle
systems (as was done earlier chiefly for the 
dipole~\cite{middel_pla,middel_cpaa,middel_anis}, but also for larger
vortex numbers in the case where rotation is present~\cite{castindum}
or absent~\cite{navar13} for co-circulating
vortices). We have shown this approach to be fairly informative 
towards an understanding of the configurations that may arise
for small vortex numbers $N$ and the identification of their stability
characteristics. Moreover, the systematic progression towards higher
vortex numbers $N$ has enabled us to extract generalizations of
the conclusions obtained for lower vortex numbers. In some cases
(e.g. for the non-existence of $N$-gon's for $N$ odd, or for the
stability characteristics of $N$-gons with $N$ even), we have been
able to prove relevant conclusions in their full generality. In other
cases (as e.g. for the collinear configurations), we have formulated
general conjectures that may, in turn, stimulate non-trivial connections
with other areas of mathematics such as algebra/number theory.
These regard, for instance, the
existence of solutions  mapping into the existence of a real point
of a certain variety and the analysis of the properties of the
corresponding near-circulant stability matrices. A deeper cross-pollinating
view that may address such open questions would certainly be a welcome
addition to the literature in the near future.

However, there are additional extensions or generalizations of the
questions posed herein that merit future investigation in their own
right. On the one hand, here we have restricted (for reasons explained)
our attention
to the case of counter-rotating vortices that may produce fixed point
configurations. However, as recent experiments have naturally argued,
it is of particular interest to also explore co-circulating vortex
states and especially rigidly rotating examples thereof (where all the
vortices rotate with the same angular momentum), rendering the co-rotating
frame of reference the right one for seeking stationary states of
the system. On the other hand, a natural generalization of the present
considerations is that of exploring the dynamics of vortex rings
in three-dimensional BECs; see e.g.~\cite{emergent,komineas} for 
relevant reviews. In this context, it is also possible to
write ordinary differential equations characterizing the
interaction of the rings and their intrinsic translational
dynamics as e.g. in~\cite{konsta}. However,
we have not been able to identify simple ODEs that would describe
the motion of such rings in a three-dimensional parabolic trap -- a key
ingredient for the system of ODEs, as we saw above for the case
of vortices. The exploration of such vortex rings as interacting
particle systems is emerging as an extremely interesting topic for
future work and will be deferred for future publications.

\vspace{5mm}

{\it Acknowledgments.} We thank C.E. Wayne for 
discussions at the early stage of this work
and R. Carretero-Gonz{\'a}lez for numerous iterations
(including on numerical results of~\cite{middel_physd}). We also
thank F. Hajir for bringing to our attention algebraic connections
of the existence problems formulated herein. PGK is grateful to the
IMA for its hospitality during the completion of this work and also
acknowledges support from NSF-DMS-0806762, NSF-CMMI-1000337 and
the US AFOSR via award FA9550-12-1-0332, as well as the
Binational Science Foundation under grant BSF-2010239.

\end{document}